\documentclass{cccg22}
\usepackage{graphicx,amssymb,amsmath}

\newcommand{\dist}{{D(\mathcal{T}, \mathcal{H})}}
\newcommand{\distt}[1]{{D(\mathcal{T}, \mathcal{H})}_{(#1)}}
\newcommand{\A}{{\mathcal A}}
\newcommand{\etal}{\textit{et al.}}
\newcommand{\pushr}[1]{\ifmeasuring@#1\else\omit$\displaystyle#1$\ignorespaces\fi}

\usepackage{xcolor}
\usepackage{algorithm,algorithmicx,algpseudocode}
\usepackage{amsmath}
\usepackage{mathtools}
\usepackage{caption}
\usepackage{enumerate}
\usepackage[hidelinks]{hyperref}

\algrenewcommand{\algorithmiccomment}[1]{/*{#1}*/}

\title{Optimally Tracking Labels on an Evolving Tree}

\author{Aditya Acharya\thanks{Department of Computer Science,
        University of Maryland,College Park MD, USA, {\tt adach@umd.edu}}
        \and
        David M. Mount\thanks{Department of Computer Science and Institute for Advanced Computer Studies, University of Maryland, College Park MD, USA, {\tt  mount@umd.edu}}}

\index{Acharya, Aditya}
\index{Mount, David M.}

\begin{document}
\thispagestyle{empty}
\maketitle

\begin{abstract}
Motivated by the problem of maintaining data structures for a large sets of points that are evolving over the course of time, we consider the problem of maintaining a set of labels assigned to the vertices of a tree. We study the problem in the evolving data framework, where labels continuously change over time due to the action of an agent called the evolver. An algorithm, which can only track these changes by explicitly probing the individual vertices, is tasked with maintaining an approximate sketch of the underlying tree. Such a framework necessitates an algorithm which is fast enough to keep up with the changes, while simultaneously being accurate enough to maintain a close approximation. We present an algorithm that allows for both randomized and adversarial evolution of the data, subject to allowing a constant speedup factor over the evolver. We show that in the limit, it is possible to maintain labels to within an average distance of $O(1)$ of their actual locations. We also present nearly matching lower bounds, both on the distance, and the speed-up factor.
\end{abstract}

\section{Introduction} \label{sec:intro}

Many modern data sets are characterized by two qualities: massive size and dynamic variation with time. The combination of size and dynamics makes maintaining them extremely challenging. Algorithms that recompute the structure can be prohibitively expensive, owing to scale of the data set. Standard models for dynamic structures (e.g., \cite{dynamic_graphs}) may not be applicable because we may not know where or when changes occur within the structure. These qualities together challenge the traditional single-input/single-output model used in the field of algorithm design.

Anagnostopoulos {\etal}~\cite{evolving_sorting_first} proposed the \emph{evolving data framework} to capture the salient aspects of such data sets. In this framework, the structure varies continuously through the actions of an \emph{evolver}, which makes small, random changes to the structure behind the scenes. Instead of taking a single input and producing a single output, an algorithm judiciously \emph{probes} the current state of the structure and attempts to continuously maintain a view of the structure that is as close as possible to its actual state. 

In this paper, we consider the problem of maintaining a tree with $n$ distinct labeled nodes in this framework. The tree topology is assumed to be fixed over time, but the evolver changes label locations by swapping the labels of two adjacent vertices. We consider the problem both in the classical evolving framework, where swaps are chosen uniformly at random, and an adversarial framework, where the evolver's swaps are arbitrary. To probe the structure's current state, we assume the existence of an \emph{oracle}, which given a pair consisting of a label and a vertex, either reports that the label truly resides at this vertex, or it returns an edge incident to the vertex indicating the first edge along the path leading from the probed vertex to the vertex where the label currently resides. 

We model our current state by means of a \emph{hypothesized labeling}, that is, a mapping of labels to the vertices. Unlike the actual labeling, the mapping need not be 1--1. Our update algorithm is extremely simple. With each step, it queries a label-vertex pair. If the label is not at this vertex, it moves the label hypothesis one vertex closer to its actual location in the tree. To measure how close our hypothesis is to the truth, we define a \emph{distance function}, which is just the sum of distances over all the labels between their hypothesized and true locations. Note that the evolver moves two labels with each step, while our algorithm moves only one. For this reason we provide our algorithm with a \emph{speedup factor} $c \geq 1$ (not necessarily an integer), which allows our algorithm to perform multiple steps for each single action of the evolver. (Further details are given in Section~\ref{sec:formulation}.)

We present four main results. We first show that, even in the most benign case of a uniform random evolver and any constant speedup, the steady-state distance over a bounded degree tree is $\Omega(n)$ (Theorem~\ref{thm:lower_bound}). Second, we show that given a speedup factor of $c = 2$ and a uniform random evolver, there exists a simple algorithm that achieves a steady-state distance of $O(n)$, for any bounded degree tree (Theorem~\ref{thm:fin1}). Next, we show that given a speedup factor of $c > 2$, for any evolver, the same simple algorithm achieves a steady-state distance of $O(n)$ (Theorem~\ref{thm:fin2}). Finally we show that for any speedup $c < 2$, there exists a tree, and an adversarial evolver, such that the steady state distance is not in $o(n^2)$ (Theorem~\ref{thm:twominus}).

\subsection{Related Work}

The problem we consider here falls under the general category of pebble motion problems. Given a graph $G = \{V,E\}$, a set of labels $L = \{l_1,l_2, \ldots, l_k\}$, a labeling configuration is defined as a mapping $M: L \rightarrow V$, such that $M(l_i) \neq M(l_j)$, for $l_i \neq l_j$. A single move of a label $l$, where $M(l) = v$, can be defined as updating the mapping to $M(l) = u$, where $u$ is a neighbor of $v$.

Given two such label assignments $M_1$ and $M_2$ on a common graph, the problem of deciding whether there is a sequence of moves to transform $M_1$ to $M_2$, was first referred to as the \textit{pebble motion problem} by Kornhauser {\etal}~\cite{pebble_motion}. Under the restriction that a label can only be moved to an unmapped neighboring vertex, Goraly and Hassin {\etal}~\cite{pebble_feasibility} show that the feasibility problem can be decided in linear time. Ratner {\etal}~\cite{pebble_nphard} proved that the associated problem of finding the optimal sequence of moves is NP-hard. 
  
A variant of pebble motion that is more closely related to this paper is the problem of \textit{token swapping}. Again we have a graph with $n$ vertices, and there are $n$ distinct labels. A single move involves swapping the labels of two neighboring vertices. It is easy to see that on a simple path, transforming one configuration to another is akin to sorting the path, and therefore such a sequence of swaps can be generated by a variant of bubble sort. Yamanaka {\etal}~\cite{swapping_labels} showed that there exists a polynomial time 2-approximation when the graph is tree. Miltzow {\etal}~\cite{token4apx} generalized this to a polynomial time 4-approximation on general graphs. Graf considered a very similar problem of moving objects along a tree by a robot and presents an excellent collection of similar problems \cite[Section~6]{sort_walking}.
  
Another related line of work involves algorithms for evolving data sets, which was first introduced by Anagnostopoulos {\etal}~\cite{evolving_sorting_first}. In their framework, the input data set is constantly changing through the actions of a random evolving agent, or \emph{evolver}, and an algorithm is tasked with maintaining an output that is close to the one corresponding to the current data. The algorithm can only access the data set through a series of probes, each of which returns some relevant local information. They considered the problem of maintaining a sorted order of points, where the true ranking of points evolves over time. Besa {\etal}~\cite{evolving_sorting} gave an optimal algorithm that maintains an approximate ordering with only $O(n)$ inversions. They showed that a repeated run of an $O(n^2)$ time sorting algorithm like the \textit{insertion sort} suffices. 
 
Researchers have considered other problems in the evolving context, including path connectivity, minimum spanning trees~\cite{evolving_graph}, shortest paths~\cite{evolving_shortestpath}, and page rank~\cite{page_rank}, among others. A common theme across these papers is the evolution of the list of edges of the graph, either through introducing a new edge, and deleting an existing one, or by changing the ranking of the edge weights.

\subsection{A New Framework for Evolving Data}

Our framework differs from the standard evolving data framework in few significant aspects. The first involves the behavior of the evolver. An important characteristic of the evolving model introduced in \cite{evolving_sorting_first} is that the evolver acts randomly, and algorithms in this model exploit the fact that the evolver will occasionally improve matters. In this paper we consider both uniformly random evolvers as well as evolvers that are non-uniform, possibly deterministic, which may act in an adversarial manner. 

The second difference is that our structure is more general in that the mapping of labels to vertices need not be 1--1. We think of the structure that the evolver acts on as a ``real world'' object, which has capacity constraints on the number of labels each vertex can hold. In contrast, we think of our hypothesized labeled point set as a theoretical model of this real-world structure, which is not constrained by real-world limitations. We also provide our algorithm with a constant speed-up factor, to handle cases when each step of the evolver effects a bigger change than that of the algorithm. In compensation for this asymmetry, our algorithms and analyses are much simpler.

The final difference is the nature of the oracle. 
We can view our problem as a generalization of evolutionary sorting, but where the domain is a tree structure, rather than a linear list. In sorting, the oracle determines whether two objects are out of order, but this is not really meaningful in our tree-based setting. Instead, our oracle provides a directional pointer to the current location of the label.


\section{Problem Formulation} \label{sec:formulation}

In this section we provide the specifics of our evolving token/label swapping problem. We are given a fixed undirected tree $T = (V,E)$ with $n$ vertices and maximum degree $k$. Each vertex of the tree is assigned a unique label from the set of labels $L = \{l_1, \ldots, l_n\}$, that is, there is a bijective mapping $M_T : L \rightarrow V$. At any time, let $\mathcal{T} = \{T,M_T\}$ denote the current ``true'' labeled tree (see Figure~\ref{fig:Example}(a)).
 
\begin{figure}[h]
    {\includegraphics[width = \linewidth]{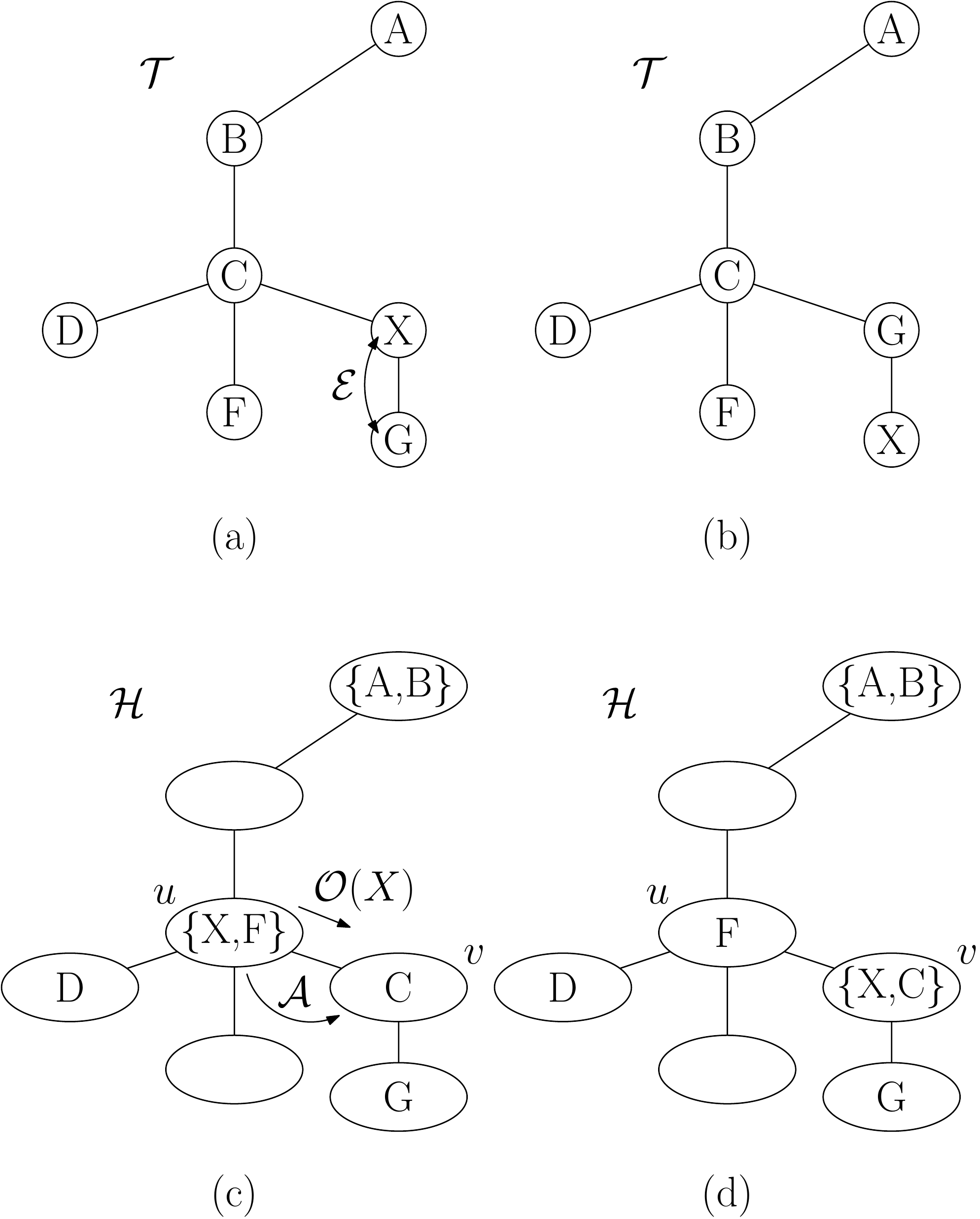}}
    \captionsetup{singlelinecheck=off}
    
    \caption[foo]{The action of the algorithm on a labeled tree $\mathcal{T}$, evolver $\mathcal{E}$, a labeled hypothesis tree $\mathcal{H}$, and oracle $\mathcal{O}$.
    (a): The current state of the underlying labeled tree $\mathcal T$. 
    (b): The state of $\mathcal T$ after the evolver swapped labels across a pair of adjacent nodes. 
    (c): A single step in our algorithm $\mathcal{A}$ on $\mathcal H$---Query label $X$, find that the oracle is pointing us to the location of $X$ on $\mathcal{T}$, and then move the label $X$ to the adjacent node in the returned direction. 
    (d): The final state of our hypothesis tree $\mathcal{H}$ after a single step of $\mathcal A$.}
\label{fig:Example}
\end{figure} 

The \textit{evolver}, denoted $\mathcal{E}$, introduces changes to the labelings. Each time it runs it selects a pair of adjacent vertices in $T$ and swaps their labels. The evolver may either be \emph{random} or \emph{adversarial}. In the former case the pair to be swapped is chosen uniformly at random, and in the latter the adjacent pair can be chosen arbitrarily, deterministically or adversarially. In Figure~\ref{fig:Example}(a) and (b), the evolver swaps labels $X$ and $G$.
 
Our algorithm maintains a model of current labeled tree in the form of a structure we call a \emph{hypothesis tree}, denoted $\mathcal{H} = \{T, M_H\}$, where $T$ is the same tree, and $M_H : L \rightarrow V $ is a (not necessarily bijective) mapping from labels to vertices. Note that $M_H$ may assign multiple labels to a vertex of $T$ (see Figure~\ref{fig:Example}(c)).

In order to probe the current actual state, we assume the existence of \emph{oracle}, denoted $\mathcal{O}$. Each query to the oracle is presented in the form of a pair $(l_i, u)$, where $l_i$ is a label and $u$ is a vertex. If $l_i$ is currently located at $u$, the oracle returns a special value \textit{null}. Otherwise, it returns the edge incident to $u$ that lies on the shortest path from $u$ to $M_T(l_i)$, the vertex that contains $l_i$ in the true labeling. (In Figure~\ref{fig:Example}(c), the query $\mathcal{O}(\textrm{X}, u)$ returns the edge $(u,v)$ because in the actual tree, the path to the node $w$ containing $\textrm{X}$ contains this edge.)

Each single step of algorithm $\mathcal{A}$ involves the following actions: $\mathcal A$ selects a label $l$ and a vertex $u$. Then queries the oracle to find $\mathcal{O}(l, u)$ and then is free to move the label $l$ from $M_H(l)$ to any adjoining node in the tree. A step of one such algorithm is illustrated in Figure~\ref{fig:Example}(c) and (d), where the algorithm is applied to label $\textrm{X}$. The query $\mathcal{O}(\textrm{X}, u)$ returns $(u,v)$, and the algorithm moves label $\textrm{X}$ to $v$. We define $\mathcal C$ as the class of such algorithms, and throughout this paper we only consider algorithms from this class. \label{def:class}

To measure how close our hypothesized labeling is to the true labeling we introduce a natural \emph{distance function}. Given two vertices $u$ and $v$ in $T$, define their distance $d(u,v) = d_T(u,v)$ to be the tree distance, i.e., the length (number of edges) of the path between them. Given the true labeling $\mathcal{T}$ and the hypothesized labeling $\mathcal{H}$ and any label $l_i$, let $D_i = d(M_T(l_i), M_H(l_i))$ denote the distance between the assigned label positions. Define the overall distance to be $D(\mathcal{T}, \mathcal{H}) = \sum_{l_i \in L} D_i$. \emph{Remark:} $\dist$ is a metric since it is the sum of tree distances, which are themselves metrics for a particular label.

Observe that with each step the evolver can affect the overall distance by at most 2, moving each of the labels being swapped one node farther from our current hypothesis. Since we have $n$ vertices and the maximum distance between two nodes on the tree is $n-1$, we have $\dist \in O(n^2)$. It is easy to see that there exists a tree $T$ and a sequence of swaps by the evolver, which results in $\dist \in \Omega(n^2)$. Specifically, consider the case where $T$ is a path and the labels are swapped in a sequence to result in a labeled path with the labels sorted in the opposite order. On the other hand, our algorithm clearly satisfies the following invariant: Every step of an algorithm from class $\mathcal C$ reduces the overall distance $\dist$ by at most 1.

Given the disparity between the evolver's and our algorithm's effect on $\dist$, we will allow our algorithm a modest \textit{speedup factor}. We denote this by a constant $c \geq 1$. This means that the time taken by a single step of the evolver is $c$ times as that of the algorithm. Or in other words, over a large enough time interval if the algorithm takes $m$ steps, the evolver takes $m/c$ steps. 

The problem considered for a given speedup factor $c$ and any arbitrary starting configuration of $\mathcal H$: Does there exist an algorithm with this speedup factor such that, in the steady state, after arbitrarily long execution sequences, $\dist = o(n^2)$? We will in fact show that (depending on the nature of the evolver) that there exists a deterministic algorithm and an associated speedup factor such that $\dist = O(n)$ from the underlying labeled tree, after some sufficiently large time and with high probability.

\section{Probabilistic Tools}
 
In this section we mention the probabilistic tools we use through out the paper. First, as a concentration bound, we use a weak version of Chernoff's inequality (Theorem~{4.5} in~\cite{book_probability}).
 
\begin{lemma}[Chernoff Bound] \label{lem:ch}
Let $X_1, X_2, \ldots, X_n$ be independent random indicator variables, let $X = \sum_{i} X_i$, and let $\mu = E[X]$. Then, $\Pr\left[X \leq\frac{\mu}{2} \right] \leq \exp(-\frac{\mu}{8})$.
\end{lemma}
 
Next we use a concept called Poisson approximation. Suppose $X_1, X_2, \ldots, X_n$, are the random variables indicating the number of balls in the $i^{th}$ bin, when $m$ balls are thrown into $n$ bins uniformly at random. We call this the exact case.
 
Let $Y_1, Y_2, \ldots, Y_n$ be  independent Poisson random variables with $\Pr[Y_i = k] = e^{-\lambda} \frac{\lambda^k}{k!}$, where $\lambda = m/n$. In other words, $Y_i$ represents the load in a bin, when the number of balls in each of them is a Poisson distribution with parameter $\lambda$. We note the following on any event that is a function of the loads of each bin. (Corollary 5.9 \cite{book_probability}.) 
 
\begin{lemma}[Poisson Approximation] \label{lem:poisson}
Any event that takes place with probability $p$ in the Poisson case takes place with probability at most $p~e~\sqrt{m}$ in the exact case.
\end{lemma}

\section{Lower Bounds on the Distance}

We first prove a lower bound on $\dist$, when the maintaining algorithm is in the class $C(\mathcal A)$ as defined in Section~\ref{sec:formulation} and for any constant speedup factor $c$. Our proof follows the same structure as a similar proof by Anagnostopoulos \etal \cite{evolving_sorting_first}. We prove the following for $\distt{t}$, for a sufficiently large $t$, where $\dist_{(t)}$ denotes $\dist$ at time $t$. 

\begin{theorem} \label{thm:lower_bound}
For any speedup factor $c \geq 1$ and for all sufficiently large $t$, irrespective of the algorithm $\mathcal{A}$, $\distt{t} = \Omega(n)$ with high probability, even in the case of a random evolver.
\end{theorem}

\begin{proof}
For ease of analysis we let our algorithm $\mathcal A$ run a single step every time unit, and the evolver, which runs $c$ times slower perform a swap every $c$ time units. Consider the time interval $[t-n/w, t]$, where $w$ is a large constant. The algorithm and the evolver can reduce $\dist$ by at most $n/w$ and $2n/cw$ during this time interval, respectively. So if $\distt{t-n/w}$ was at least $n/w + 2n/cw + \Omega(n)$, then $\distt{t}$ remains $\Omega (n)$.

Next, let us assume $\distt{t - n/w}$ is at most $n/w + 2n/cw + o(n)$. That implies there are at most $n/w + 2n/cw + o(n)$ labels displaced from their true location at time $t-n/w$. Let $L'$ denote the set of displaced labels, that is, $L' = \{l_i \mid D_i > 0\}$. We define $V' = \{M_T(l_i) \mid l_i \in L'\}$, as the set of corresponding vertices on $T$. And then the set of incident edges as $E' = \{(u,v)\mid u \in V' \lor v \in V'\}$. Since the degree of the $T$ is $k$, we have $|E'| \leq k(n/w + 2n/cw + o(n))$.

In the same time frame, the algorithm $\mathcal{A}$ can act on at most $n/w$ labels. Call that set of labels $L_\mathcal{A}$. Define $V_\mathcal{A} = \{M_T(l_i) \mid l_i \in l_\mathcal{A}\}$, as the set of corresponding vertices on $T$. And then the set of incident edges as $E_\mathcal{A} = \{(u,v)\mid u \in V_\mathcal{A} \lor v \in V_\mathcal{A}\}$. Now, $|E_\mathcal{A}| \leq kn/w$.

Next we look at the set of edges that were unaltered at time, $t-n/w$, and were not affected by the algorithm throughout the time interval. Call it $E^* = E \setminus (E' \cup E_\mathcal{A})$. Now, $|E^*|$ $\geq n - 2kn/w - 2nk/cw - k\cdot o(n)$ $\geq n\gamma$, for some sufficiently large $w$, and $\gamma = (1 - 2k/w - 2k/cw - k/w)$. The evolver picking any edge from $E^*$ exactly once, guarantees that the labels stay swapped at the end of the time interval.

Let $X_e$ be the indicator variable, representing the fact that $e$ is picked by the evolver exactly once. We use the Poisson approximation scheme from Lemma~\ref{lem:poisson}. The evolver chooses $n/cw$ edges at random from the $n$ available ones. Therefore $\lambda = (n/cw)/n = 1/cw$, which is a constant. Hence, $\Pr[Y_e = 1] = \lambda e^{- \lambda} = s$, a constant. That implies, $\mathrm{E}[\sum_{e\in E^*} Y_e] \geq s\gamma n$. Using a Chernoff bound (Lemma~\ref{lem:ch}), we have
\[
    \Pr \left[\sum_{e\in E^*} Y_e \leq s\gamma n/2\right] 
        ~ \leq ~ e^{-\Omega(n)}.
\]
Using Lemma~\ref{lem:poisson} again, we have
\begin{align*}
    &\Pr \left[\sum_{e\in E^*} X_e \leq s\gamma n/2\right] \\
        &  \leq ~ e\,\sqrt{\frac{n}{cw}}\Pr \left[\sum_{e\in E^*} Y_e \leq s\gamma n/2\right]\\
        & \leq ~ e\,\sqrt{\frac{n}{cw}}e^{-\Omega(n)} 
        ~ \leq ~ e^{- \Omega(n)}
\end{align*}

Therefore with exponentially high probability, the evolver picks at least $s\gamma n/2$ edges from $E^*$, ensuring that those edges stay swapped at the end of the interval. Therefore, $\distt{t} \geq s\gamma n \in \Omega(n)$, as desired.
\end{proof}

\section{Algorithm}

Here, we describe a simple algorithm to track the labels. We use the same algorithm in both the cases of a random and an adversarial evolver. Recall the set of labels $L = \{l_1, \ldots, l_n\}$, and the definition of the oracle from Section~\ref{sec:formulation}. 

Intuitively, the algorithm works as follows. For each $l_i \in L$, we query the oracle on $(l_i, M_H(l_i))$ and update its location by moving it one step in the direction returned by the oracle. We keep doing this until the oracle returns \textit{null}, that is, when $l_i$ is in its true location. We then move on to the next label, repeating the process indefinitely. A single pass over all the labels is called an \emph{iteration} of the algorithm.

\begin{algorithm}[H]
\caption{Tracking Labels}\label{alg:cap}
\begin{algorithmic}
\State{}\Comment{Continuously run the algorithm}
\For{$j \gets 1,2 ,\ldots, \infty$} 
    \State{}\Comment{For every label in order}
    \For{$i \leftarrow$ 1 to $n$}
        \State{}\Comment{Until the label is in its true location}
        \While{$(\mathcal{O}(l_i, M_H(l_i)) \neq \textit{null})$} 
            \State{}\Comment{Query the oracle to find the direction}
            \State{$(u,v) \gets \mathcal{O}(l_i, M_H(l_i))$} 
            \State{}\Comment{Update the location of the label}
            \State{$M_H(l_i) \gets v$} 
            
        \EndWhile
    \EndFor
\EndFor

\end{algorithmic}
\end{algorithm}

\section{Analysis}

Again for ease of analysis, we let our algorithm $\mathcal A$ run a single step every time unit, and the evolver, which runs $c$ times slower, perform a swap every $c$ time units.

Let $t_0$ be the time when the algorithm starts. Let $t_j$ be the time when the $j$th iteration of the algorithm ends. Let $\dist$ at the start of the $j$th iteration be $\dist_j$. And for a specific label $l_i$ we denote the distance at the start of the $j$th iteration to be $D_{i,j}$.

We set the total number of moves effected on $l_i$, by the algorithm in the $j^{th}$ iteration as $\A_{i,j}$. Therefore the total decrease in $D_i$, the distance with respect to label $l_i$, in the $j^{th}$ iteration is $\A_{i,j}$. We define the total decrease in $\dist$ due to the algorithm, in the $j^{th}$ iteration as $\A_j$, $\A_j = \sum_{l_i \in L} \A_{i,j}$.

We note the following about $\Delta t_j$, the time taken by the $j^{th}$ iteration.

\begin{lemma} \label{lem:t1}
$\Delta t_j = t_j - t_{j-1} = n + \A_j$.
\end{lemma}

\begin{proof}
Every step of the algorithm either moves a label in the direction of its true location, or fixes it, i.e., finds the label is in its true location. Since there are $n$ labels, and $\A_j$ is the total moves effected by the algorithm, we have the result
\end{proof}

Next we show a lower bound for the time taken by the $j^{th}$ iteration.

\begin{lemma} \label{lem:t2}
$\Delta t_j \geq \frac{c}{2+c} (\dist_j + n)$.
\end{lemma}

\begin{proof}
For a specific label $l_i$, our algorithm reduces its distance by $\A_{i,j}$, then finds that the label is at its true location, and then moves on to the next label. This implies that for some subset of steps taken by the evolver, the distance associated with $l_i$ was reduced by $D_{i,j} - \A_{i,j}$. Otherwise, the algorithm would not have moved on to the next label.

This further implies that in the $j^{th}$ iteration for some subset of its steps, the evolver reduced the overall distance by at least $\sum_i (D_{i,j} - \A_{i,j})$ = $\dist_j - \A_j$. That takes the evolver at least $(\dist_j - \A_j)/2$ steps, or at least $(c/2)(\dist_j - \A_j)$ time.

Therefore we have $\Delta t_j \geq (c/2)(\dist_j - \A_j)$. Using Lemma~\ref{lem:t1}, we have $\Delta t_j \geq \frac{c}{2}(\dist_j - \Delta t_j + n)$. Simplifying the inequality gives us the desired result
\end{proof}

\subsection{Random Evolver and Speedup $\mathbf 2$} \label{analysis1}

In this section we prove the following: In the case of a random evolver, where the evolver $\mathcal E$ picks an edge at random and swaps its labels, an algorithm that runs at least twice as fast as the evolver maintains an optimal distance. Or in other words, we show that for $c \geq 2$, our algorithm ensures $\dist \in O(n)$ with high probability. Using Theorem~\ref{thm:lower_bound}, we can conclude that our algorithm is optimal for $c \geq 2$ and a random evolver.

As in \cite{evolving_sorting}, we first prove an interesting result about the random evolver. We show that a constant fraction of the steps taken by the random evolver do not increase the overall distance $\dist$. 

\begin{lemma} \label{lem:const}
For $c = 2$ and degree $k$, there exists a constant $\epsilon$, $0<\epsilon <1$, such that for all $j$, the random evolver does not increase the overall distance in at least $\epsilon \Delta t_j$ steps in the $j^{\text th}$ iteration, with high probability.
\end{lemma}

\begin{proof}
From Lemma~\ref{lem:t1}, we know $\Delta t_j$ is at least $n$. We look at the first $n/10k$ steps of this particular iteration. The algorithm can process at most $n/10k$ nodes in this time. The number of edges incident on these nodes is at most $n/10$. Let $E'$ denote the set of edges left unaltered by the algorithm in this time interval. Then $|E'| \geq 9n/10$. In the same time period, the evolver picks edges at random from the edge set $E$, $n/20k$ times with replacement.

For every edge $e$ in $E$, we set $X_e = 1$, if $e \in E'$, and the evolver picks $e$, at least twice in the time-frame, but picks none of the edges incident on $e$. 

We use the Poisson approximation scheme from Lemma~\ref{lem:poisson}. The evolver chooses $n/20k$ edges from the $n$ available ones. Therefore $\lambda = (n/20k)/n = 1/20k$, which is a constant. 
Now let $Y_e$ be the independent Poisson approximations of $X_e$, with $\lambda = 1/20k$. 

Next we find $\Pr[Y_e = 1]$. That represents the event when $e$ is picked from $E'$, and $e$ is picked twice but none of the edges incident on $e$ are picked. In the Poisson approximation scenario, each edge is picked $j$ times with a probability $e^{- \lambda} \lambda ^j / j!$. Therefore, the probability that an edge is picked at least twice is $(1- e^{-\lambda} - \lambda e^{-\lambda})$, and the probability that it is not picked whatsoever is $e^{-\lambda}$.  Since at most $2k$ edges can be incident on $e$, we have
\[
    \Pr[Y_e = 1] 
        ~ \geq ~ \frac{9}{10}\left(1- e^{-\lambda} - \lambda e^{-\lambda}\right)e^{-2k \lambda}.
\]

Since the right hand side is a constant, there exists $s = O(1)$ such that $\Pr[Y_e = 1] \geq s$. Therefore, $\mathrm{E}[\sum_{e\in E} Y_e] \geq s n$. Using a Chernoff bound (Lemma~\ref{lem:ch}), we have
\[
    \Pr \left[\sum_{e\in E} Y_e \leq sn/2\right] 
        ~ \leq ~ e^{-\Omega(n)}.
\]
Using Lemma~\ref{lem:poisson} again, we have
\begin{align*}
    \Pr \left[\sum_{e\in E} X_e \leq sn/2\right] 
        &  \leq ~ e\,\sqrt{\frac{n}{20k}}\Pr \left[\sum_{e\in E} Y_e \leq sn/2\right]\\
        & \leq ~ e\,\sqrt{\frac{n}{20k}}e^{-\Omega(n)} 
        ~ \leq ~ e^{- \Omega(n)}.
\end{align*}

We note that if $X_e = 1$, then $e$ is left unaltered by the algorithm, but it is altered at least twice by the evolver. That further means, one of those steps by the evolver either decreases the overall distance $\dist$ or leaves it unchanged. And since the number of such edges $e$, with $X_e = 1$, is at least $s n/2$ with exponentially high probability, we conclude that in at least $s n/2$ of the evolver steps, in the first $n/10k$ steps of the iteration, $\dist$ does not increase. Dividing the iteration into chunks of $n/10k$ steps, we obtain the desired result.
\end{proof}

Finally we prove one of the main theorems of this paper, that for a long enough passage of time, $\dist$ converges to $O(n)$, in the case of $c=2$, and a random evolver.

\begin{theorem} \label{thm:fin1}
Given a tree of size $n$ and a constant degree, and a random evolver, there exists $z$ (a function of $n$) such that for all $j > z$, Algorithm~\ref{alg:cap} achieves $\dist_{j} \in O(n)$, with a speed-up factor $c=2$.
\end{theorem}

\begin{proof}
Consider the $j^{\text th}$ iteration. From Lemma~\ref{lem:const}, the evolver increases $\dist$ by at most $(1-\epsilon)\Delta t_j$. In the same iteration the algorithm reduces $\dist$ by $\A_j$. Therefore, with high probability:
\begin{align*} 
    &\dist_{j+1} \\
        & \leq ~ \dist_{j} + (1-\epsilon)\Delta t_j - \A_j \\
        &  \leq ~ \dist_{j} + n -\epsilon\Delta t_j & \text{[Lemma~\ref{lem:t1}]}\\
        &  \leq ~ \left(1 - \frac{\epsilon}{2}\right)\dist_{j} + \left(1 - \frac{\epsilon}{2}\right)n & \tag*{\text{[$c=2$ in Lemma~\ref{lem:t2}]}}\\
        &  =    ~ \left(1 - \frac{\epsilon}{2}\right)^j \dist_0 + \sum_{s=1}^{j}\left(1 - \frac{\epsilon}{2}\right)^j n\\
        &  \leq ~ \left(1 - \frac{\epsilon}{2}\right)^j n^2 + O(n). & \tag*{\text{[since $\dist \leq n^2$}]}
\end{align*}
By choosing $z = \log _{1/(1-\epsilon /2)} n$, we have $\dist_{z+1} \in O(n)$.
\end{proof}

\emph{Remark:}
We showed that for large enough $j$, $\dist_j \in O(n)$. Can we conclude the same about $\dist$ throughout the $j^{th}$ iteration as well? In particular we look at $\dist_{i,j}$. We note that in our Algorithm~\ref{alg:cap} we could have started with processing the label $l_i$ first (instead of $l_1$), $l_{i+1}$ next, and so on. Therefore for a large enough $j$, $\dist_{i,j} \in O(n)$ as well. Since $\dist_{i+1,j} \leq \dist_{i,j} + O(n)$, we conclude that for a large enough passage of time $\dist$ converges to $O(n)$.


In our labeled hypothesis tree $\mathcal H$ multiple labels could reside at a particular node. We show a simple result on the maximum number of labels that could be mapped to single vertex in $T$.

\begin{cor}
\label{cor:max}
Let $L_{\mathcal H,v}$ be the set of labels residing at a node $v$ in $\mathcal H$, after a long enough passage of time. Then, $|L_{\mathcal H,v}| \in O(\sqrt{n})$
\end{cor}

\begin{proof}
\label{cor:max:proof}
Let $|L_{\mathcal H,v}| = w$.  For $l_i$'s, $l_i \in L_{\mathcal H,v}$, we consider the corresponding distances $D_i$'s. Consider that set as $D_v$, $D_v = \{D_i|l_i \in L_{H,v}\}$. Since the tree has degree has $k$, there can be at most $k$ 1's in $D_v$, similarly $k$ number of 2's, and so on. At most one member of $D_v$ can be zero. Therefore 
\[
    \dist 
        ~ \geq ~ \sum_{x \in D_v} x 
        ~ \geq ~ k(1 + 2 + \cdots + (w-1)/k) \in \Omega(w^2).
\]
Since $\dist \in O(n)$ after a long enough time from Theorem~\ref{thm:fin1}, we conclude $w \in O(\sqrt{n})$.
\end{proof}

\subsection{Adversarial Evolver and Speedup $\mathbf{> 2}$} \label{analysis2}

We conclude with the case when the evolver is adversarial. That means we cannot rely on a result similar to Lemma~\ref{lem:const}. We show that for a speedup factor of $c>2$, or in other words, if there exists $\delta \in \mathbb{R}$, such that $c = 2+\delta$, we can still maintain an optimal distance. 

\begin{theorem}\label{thm:fin2}
Given a tree of size $n$, an adversarial evolver, there exists $z$ (a function of $n$) such that for all $j > z$, Algorithm~\ref{alg:cap} achieves $\dist_{j} \in O(n)$, with any speed-up factor $c>2$.
\end{theorem}

\begin{proof}
Consider the $j^{\text th}$ iteration. The evolver increases $\dist$ by at most $\frac{2\Delta t_j}{c}$. Therefore
\begin{align*} 
    &\dist_{j+1} \\
    &\leq \dist_{j} + \frac{2\Delta t_j}{c} - \A_j \\
        &  =    ~ \dist_{j}  + \frac{2\Delta t_j}{c} + n - \Delta t_j \tag*{\text{[Lemma~\ref{lem:t1}]}}\\
        &  =    ~ \dist_{j} + n - \left(1-\frac{2}{c}\right)\Delta t_j.
\end{align*}
By applying Lemma~\ref{lem:t2}, we have
\begin{align*}
    &\dist_{j+1} \\
        &  \leq ~ \dist_{j} + n - \frac{c-2}{2+c} \left(\dist_j + n\right) \tag*{\text{[Lemma~\ref{lem:t2}]}}\\
        &  =    ~ \frac{4}{2+c}\dist_j + \frac{4}{2+c}n\\
        &  \leq ~ \left(\frac{4}{2+c}\right)^j \dist_0 + \sum_{s=1}^{j} \left(\frac{4}{2+c}\right)^j n\\
        &  \leq ~ \left(\frac{4}{2+c}\right)^j n^2 + O(n). \tag*{\text{[$c>2$, $\dist \leq n^2$]}}
\end{align*}
For $c > 2$ , and by choosing $j > \log _{\frac{c+2}{4}} n$, we have $\dist_{j+1} = O(n)$.
\end{proof}

\subsection{Adversarial Evolver and Speedup $ \mathbf{< 2}$}

We adapt a construction from Biniaz \etal ~\cite{tokenadversary} to prove a lower bound on the required speed-up to ensure  $\distt{t} \in O(n)$. Construct two configurations of a labeled tree $\mathcal T_0$, and $\mathcal T_1$ as in Figure~\ref{fig:adv}. On such a tree: $D(\mathcal T_1, \mathcal T_0) \sim 2~OPT$, where $OPT$ is the number of optimum swaps required to go from one configuration to the other. Intuitively, an algorithm running at a speed-up factor less than 2, will fail to catch up with an adversarial evolver that takes $OPT$ swaps to modify $\mathcal T_0$, to $\mathcal T_1$. We can show that any algorithm from class $\mathcal C$ running with speed-up $2-\delta$, where $\delta$ is a small positive constant, cannot achieve $\dist \in O(n)$. In fact we can prove something stronger:
\begin{theorem}[Lower bounds on speed-up]
\label{thm:twominus}
    Given any time instant $t_0$, there exists a tree $T$, an adversarial evolver $\mathcal E$,  and a time instant $t>t_0$ s.t. $\distt{t} \in \Omega(n^2)$, for any algorithm from class $\mathcal C$, which runs with a speedup $2-\delta$, where $\delta$ is a positive constant
\end{theorem}
\begin{proof}
    Suppose we have access to an algorithm $\mathcal A$ from the class $\mathcal C$, as defined in Section~\ref{def:class}, with a speed-up factor of $c = 2-\delta$, $\delta$ is a positive real constant. We show the existence of a tree, and an adversarial evolver, where such a speed-up is not sufficient for $\dist \in O(n)$.

We adapt a construction from Biniaz \etal ~\cite{tokenadversary}. See Figure~\ref{fig:adv}. Let $\mathcal T_0$ be a uniquely labeled tree, with $\beta$ wings, $\alpha$ tails, and a central vertex. Each wing contains $\alpha$ nodes. For our purposes, we let $\alpha \in \Omega (n)$. $n = \alpha \beta + \alpha + 1$. Let $\mathcal T_1$ be another labeled instance of the same tree, where the labels of the wings, are cyclically permuted. The order of the labels on a wing remains the same, as do other labels of the tree. This gives us $D(\mathcal T_1, \mathcal T_0) = \beta \alpha (\alpha +1)$.
\begin{figure}

    \centerline{\includegraphics[scale=0.60]{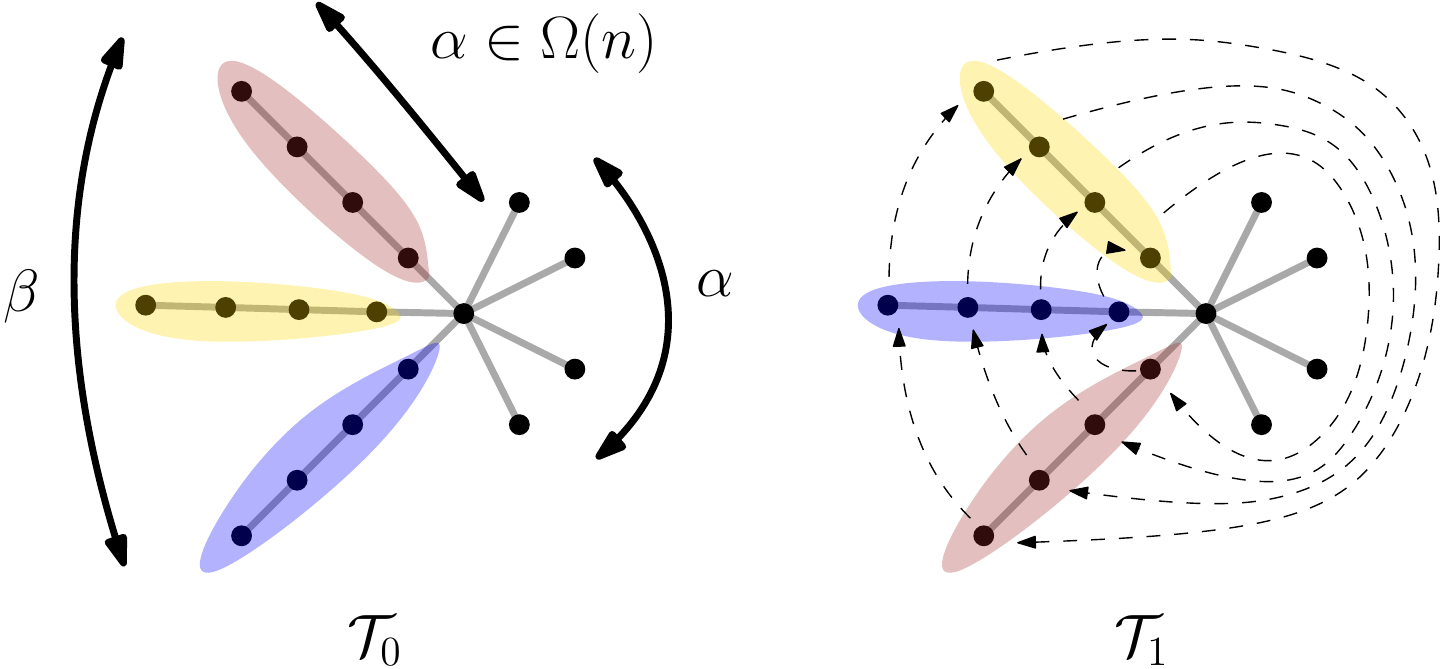}}
    \caption{$\mathcal T_0$ is a $n$-node tree with $\beta$ wings, $\alpha$ tails, and a central vertex. Each wing contains $\alpha$ nodes. $\mathcal T_1$ has the labels of the wings of $\mathcal T_0$ cyclically permuted. Adapted from~\cite{tokenadversary}. }
    \label{fig:adv}
\end{figure}

Biniaz \etal ~\cite{tokenadversary} show that the optimal number of adjacent swaps to go from $\mathcal T_0$ to $\mathcal T_1$ is $opt(\alpha, \beta) = {(\beta + 1)(\alpha (\alpha +1)/2 + 2\alpha)}$. Consider a time $t_0$, where $\mathcal T_0$ is the labeled configuration of the tree, with our hypothesis tree $\mathcal H_0$ being exact, i.e., $D(\mathcal T_0, \mathcal H_0) = 0$. Next, consider an adversarial evolver $\mathcal E$, which performs $opt(\alpha,\beta)$ number of swaps such that at time $t_1 = t_0 + opt(\alpha,\beta)$, $\mathcal T_1$ is the true labeling.

Let $\mathcal H$ be the hypothesized labeling at time $t_1$. Since $\mathcal A$ has a speed-up of $2-\delta$, and can affect the distance by at most 1 every step, we have $D (\mathcal H_1, \mathcal T_0) \leq {(2-\delta)~opt(\alpha, \beta)}$.
Considering $\beta = 2/\delta$, and $\alpha = \Omega(n)$ we have the following:

\begin{align*}
    &D(\mathcal H_1, \mathcal T_1) \\
    &\geq D(\mathcal T_1, \mathcal T_0) - D(\mathcal H_1, \mathcal T_0) \tag*{{[$D(\cdot,\cdot)$ is a metric]}}\\
    &\geq \beta \alpha (\alpha +1) - (2-\delta)~(\beta + 1)\left(\frac{\alpha(\alpha +1)}{2} + 2\alpha\right)\\
    & \geq \left(\frac{2}{\delta} - \frac{(2-\delta)(2+\delta)}{2\delta}\right)\alpha (\alpha +1) - s(\delta) \alpha \tag*{[Set $\beta = \frac{2}{\delta}$, $~s(\delta)$ is a constant]}\\
    &\geq \frac{\delta}{2}\alpha (\alpha +1) - s(\delta) \alpha ~\in~ \Omega(n^2). \tag*{[For $\alpha \in \Omega (n)$, and constant $\delta$]}
\end{align*}

\end{proof}

\section{Concluding Remarks}

In this paper, we have presented an efficient algorithm for tracking vertex labels in a tree in the evolving data framework. Our algorithm allows for both randomized and adversarial evolution of the data, subject to allowing a constant speedup factor over the evolver. Our analysis showed that in the limit, it is possible to maintain labels to within an average distance of $O(1)$ of their actual locations. We also presented nearly matching lower bounds, both on the distance and the speed-up factor.

This raises the question whether the evolving data framework can be fruitfully applied to tracking the movement of objects through more complex spaces and structures. Applications include real-time tracking of moving agents through GPS tracking of unmanned aerial vehicles \cite{uav} and tracking disease hot-spots that evolve over the course of time \cite{hotspot}.

We would like to thank Michael Goodrich for introducing us to the evolving data framework and for inspiring discussions on this topic.

\bibliographystyle{plain}
\bibliography{evolving,udg,cccg}
\end{document}